\documentclass{acm_proc_article-sp}   
\usepackage{graphicx,url}
\usepackage{enumerate}
\usepackage[ruled,section]{algorithm}
\usepackage{algorithmic}
\newcommand{\algrdRat}{RatLU}
\newcommand{\algrdDet}{PrecDetLU}
\newcommand{\algrdMat}{PrecMatLU}
\newcommand{\algrdDixon}{PrecMatDixon}

\newtheorem{theorem}{Theorem}[section]
\newtheorem{lemma}[theorem]{Lemma}

\title{Towards an exact adaptive algorithm for the determinant of a rational matrix}
\author{Anna Urba\'nska\\
Laboratoire Jean Kuntzmann\\
Universit\'e Joseph Fourier, Grenoble I\\
\texttt{E.mail: Anna.Urbanska@imag.fr}\\
}

\date{}

\DeclareMathOperator{\mmod}{\operatorname{mod}}
\DeclareMathOperator{\diag}{\operatorname{diag}}
\DeclareMathOperator{\EEA}{\operatorname{EEA}}

\begin{document}

\maketitle

\begin{abstract}
In this paper we propose several strategies for the exact computation
of the determinant of a rational matrix. First, we use the Chinese
Remaindering Theorem and the rational reconstruction to recover the
rational determinant from its modular images. Then we show a
preconditioning for the determinant which allows us to skip the
rational reconstruction process and reconstruct an integer result. We
compare those approaches with matrix preconditioning which allow us to
treat integer instead of rational matrices. This allows us to introduce
integer determinant algorithms to the rational determinant problem. In particular, we discuss the applicability of the adaptive determinant algorithm of \cite{jgd:2006:det} and compare it with the integer Chinese Remaindering scheme.
We present an analysis of the complexity of the strategies and
evaluate their experimental performance on numerous examples. This
experience allows us to develop an adaptive strategy which would
choose the best solution at the run time, depending on matrix
properties.
 All strategies have been implemented in LinBox linear algebra
library.
\end{abstract}

\section{Introduction}
The determinant computation is one of the core problems in linear algebra.
To our knowledge, the problem of the exact computation of the
determinant of a rational matrix (i.e a matrix with rational
entries) has not yet been widely studied. In general, exact
algorithms can be used everywhere where large precision is required.
For example, the determinant can be too close to 0 or $\pm \infty$
and thus cannot be computed by floating point precision algorithms.
In the case of ill-conditioned matrices symbolic methods can be
preferred as rounding errors can spoil the computation. It can also
be interesting to compare the use of decimal and continued fractions
approximations of the entries of real-valued matrices. Continued fractions are
the best approximants with small denominators, see \cite[Ch.
4]{gathen}. In this paper, we will try to face the question of how
efficient an exact determinant computation can be in both cases.

LinBox library \cite{Dumas2002Linbox} implements exact algorithms for the determinant computation in the case of modular and integer domains. By using fast modular routines \cite{Dumas2002issac,Dumas2004} it can offer solutions an order of magnitude faster than other existing implementations \cite{jgd:2006:det}. We apply these procedures to the computation of the determinant of a rational matrix.

Rational field arithmetics is implemented in GMP \cite{GMP} and
Givaro \cite{Givaro} libraries. In general, rational numbers are
difficult to treat from the exact computation point of view. Mainly,
the size of the numerator and denominator can increase very quickly
with every addition and multiplication. When we add or multiply two
fractions with numerators and denominators bounded by $M$, the
numerator and denominator of the result are bounded by $O(M^2)$.
Moreover, one addition requires 3, and one multiplication requires 2
integer products, as well as a $\gcd$ computation. Therefore, the
cost of an exact matrix-vector or matrix-matrix product can be
prohibitive in practice.
This prohibits the use of the rational field $\mathbb{Q}$ in most exact linear algebra algorithms which rely on matrix-matrix or matrix-vector products.

However, the cost of computing a modular image of a rational number $\frac{a}{b}$, where $a, b$ are of moderate size, should be comparable with the cost of computing a modular image of a large integer number. This allows us to compute a modular image of a rational matrix at a reasonable cost and thus enables us to use modular procedures. 

To compute the determinant of a rational matrix
$A=\big[\frac{a_{ij}}{b_{ij}}\big], b_{ij}>0$ the problem of matrix
storage has to be considered. First, we can store the entries of $A$
as rational numbers. Furthermore, one could store the common
denominator $D(A)$ of all entries of $A$ and an integer matrix $A'$
given by the formula $A=\frac{1}{D(A)}A'$. This approach can be
useful in the case when the entries of $A$ are decimal fractions and
$D(A)$ can be set to a power of $10$. But if we only assume that the
values $|a_{ij}|, b_{ij}$ are less than $M$, both $D(A)$ and
$\|A'\|$ are bounded by $O(M^{n^2})$. Still, we may store the common
denominator for each row (column) separately. Then the integer
vectors $\tilde{A}_i$ are given by the equation
$A_i=\frac{1}{D_i}\tilde{A}_i$, where $A_i$ is the matrix row
(column) and $D_i$ is the common denominator of its entries. Vectors
$\tilde{A}_i$ form matrix $\tilde{A}$, the norm of which is bounded
by $O(M^{n})$. The product $\pi D_i$ gives a more accurate
approximation of the denominator $D(\det(A))$ than $D(A)^n$.

The purpose of this paper is to propose the strategies to compute the
denominator of a rational matrix. All approaches are based on modular
computation. Depending on the matrix storage determinant and/or matrix
preconditioning is proposed. The resulting algorithms can use the
rational reconstruction \cite[Ch.5]{gathen} and/or existing integer determinant algorithms.

The rest of the paper is organized as follows. In section \ref{sec:ea} we give a short description of the existing algorithms for the rational reconstruction and the integer determinant problem. In section \ref{sec:rda} we present the main result i.e. two preconditioning strategies and four new algorithms to compute the rational determinant. The cost of the algorithms can be described in terms of the number of modular images of $A$ and modular determinant computations needed. Depending on the strategy, the cost of the rational reconstruction or $p$-adic lifting is taken into account. In section \ref{sec:cpl} we discuss the cost of computing a modular image of a matrix and the overall cost of the algorithms. In section \ref{sec:exp} we present the experimental results and discuss the best choice of the strategy in practice. We conclude the paper by proposing some mixed solutions in section \ref{sec:ccl}.

\section{Existing Algorithms}\label{sec:ea}

The aim of this section is to introduce the algorithms that will be
used later in section \ref{sec:rda}. In subsection \ref{sec:rr} we
give a short description of the rational reconstruction procedure.
On the example of $p$-adic system solving \cite{Dixon1982}, we
present the application of this procedure to the computation of a
rational solution. We show how to change the procedure in the case
of early terminated reconstruction \cite{Wang} and give the
complexity estimation in this case. Then in subsection \ref{sec:ida}
we present the classical CRA algorithm for the determinant and its
modifications by \cite{Abbott1999} and \cite{jgd:2006:det}.


\subsection{Rational reconstruction and its application}\label{sec:rr}

A modular image of a rational number $\frac{a}{b}$ mod $M$ can be
computed by taking the modular images of $a$ and $b$ and applying
the modular division. This fact can be written as
\begin{equation*}
\frac{a}{b} = u  \mmod M \Leftrightarrow a=bu \mmod  M.
\end{equation*}
It should be noticed is that the opposite procedure can also be
performed. One can reconstruct the fraction $\frac{a}{b}$ where
$\gcd(a,b)=1, b>0$ from it modular image $u$. The solution is
usually not unique but when we additionally require that $|a| <
\frac{N}{2}$, $b \leq \frac{M}{N}$, then there exists at most one
solution, see \cite[Ch.5]{gathen}.


The solution to the rational reconstruction problem can be computed by
applying the extended Euclidean algorithm EEA which searches for the
$\gcd$ of $M$ and $u$. The procedure Ratrec(a,b,u,M, N, D) takes as the input modulus $M$, $u\in \mathbb{Z}$ and the bounds $N$ and $D$, and returns a fraction $\frac{a}{b}=u \mmod M$ such that $|a| <N, b < D$ or FAIL if no such solution exists.
The worst case complexity of {\em Ratrec} is thus the same as for the
EEA algorithm i.e. $O\left(\log^2(M)\right)$ for the classical
algorithm and $O(\log(M)\log(\log(M)))$ for the fast Euclidean
algorithm, see \cite[Ch.11]{gathen}. We will use the notation $\EEA(M)$ for the complexity of the Extended Euclidean Algorithm with entries bounded by $M$.

In many application, the cost of rational reconstruction is usually
small compared with the cost of computing $u$ and $M$. The general scheme is to recursively compute $u_k, M_k$, where $M_k=p_1p_2\cdots p_k$ or $M_k =p^k$ until $M_k > 2ND$ and then to apply the rational reconstruction. The complexity of the procedure depends on the number $k$ of steps, which can be quite large
Reducing the number of steps can be the easiest way to enhance the
performance of the algorithm.

This can be seen on the example of the Dixon algorithm
\cite{Dixon1982} to solve a linear system $Ax=b$ of integer
equations. Let $N,D$ be the bound for the numerator and denominator
of $x$. In the classical approach we compute the $p$-adic
approximation in $k > \log(N)+ \log(D)+1 $ steps and then
reconstruct the result, which gives the complexity
$O\left(m^3(\log(m)+\log(\|A\|))^2\right)$ when we use the bound of
Hadamard for $D,N$ and assume $b\in O(1)$.
See \cite{Mulders1999} for a detailed complexity study. In fact, the number of entries in $x$ which we need to reconstruct can often be reduced, see \cite{Dumas2002Linbox}.


One should however notice, that the bounds $N$ and $D$ can be much
bigger than the actual result.
The idea is therefore to apply the rational reconstruction
periodically and check the solution for correctness.
If $M_k=p^k$ is the modulus in the current step, the method of Wang
\cite{Wang} prompts us to set $\sqrt{\frac{M_k}{2}}$ as the current
bound for numerators and the denominator in {\em Ratrec}. The
algorithm is guaranteed to return the result if $M_k > 2\max(N(x)^2,
D(x)^2)$, where $N(x),D(x)$ are the values of the numerator and the
denominator. In the opposite case,
$Ratrec(a,b,u,M_k,\sqrt{\frac{M_k}{2}},\sqrt{\frac{M_k}{2}})$ should
fail with large probability.
If we apply Wang's idea to the $p$-adic lifting we can reduce the
number of steps to $k=2\log_p(\max(N(x),D(x)))+1$ and the complexity becomes $O\left(m^{\omega} + m^2k\log(m\|A\|)\right)$
Current work on this field focus on further reducing the number of steps in the case when $N(x) \ll D(x)$ or $D(x) \ll N(x)$. A purely heuristic idea is to use the bounds $\sqrt{\frac{M_k}{2}\frac{N}{D}}$, $\sqrt{\frac{M_k}{2}\frac{D}{N}}$ instead of $\sqrt{\frac{M_k}{2}}$. For other approaches, see \cite{kho-mon:2006,ol-sto:2006}.


\subsection{Integer Determinant Algorithms}\label{sec:ida}

For an integer matrix $A$ one has several alternatives to compute the
determinant.
The classical approach is to use Chinese Remaindering Algorithm (CRA)
to reconstruct the value from sufficiently many modular images.
The modular determinant is computed  by LU factorization in the time
$O(n^\omega)$, where $n$ is the matrix dimension.
Each step of the algorithm consist of computation mod $p_i$ and a
reconstruction of the determinant mod $p_1\cdots p_i$ by the Chinese
Remaindering Theorem. The computation is stopped when the early
termination (ET) condition is fulfilled i.e. the reconstructed
result rests the same for several iterations. The algorithm is Monte
Carlo type, where the probability of success is controlled by the
number of repetitions. See \cite{Dumas2001, jgd:2006:det} for a
detailed description.


A mixture of CRA loop and Dixon $p$-adic lifting is used to compute the integer determinant in  \cite{Abbott1999} and in the hybrid algorithm of \cite{jgd:2006:det}. The principle is to reduce the value reconstructed by CRA algorithm by computing a large fraction of the determinant. By solving several linear systems we can compute some largest invariant factors $s_{m},\dots s_{m-i}$. Their product $\pi$ is potentially a large part of the determinant. An early terminated CRA loop which reconstructs $\det(A)/\pi$ mod $p_0p_1\cdots p_i$ usually requires only a few modular determinant computations. Informally, the algorithm can be described as follows.
\begin{enumerate}
\item For $i=0$ to $k$ do
\begin{enumerate}
\item Solve $Ax_i = b_i$ by Dixon $p$-adic lifting to find $s_{m},\dots s_{m-i}$;
\item $\pi = s_{m}\cdots s_{m-i}$;
\item Run CRA for several iterations to determine $\det(A)/\pi$;
\item if ET break;
\end{enumerate}
\item Run another determinant algorithm to get the result;
\end{enumerate}
Here, $k$ should not exceed the expected number of invariant factors which is $O(\sqrt{\log(n)})$ see \cite{jgd:2006:det}.
The expected complexity of the hybrid determinant algorithm
\cite{jgd:2006:det} for random dense matrices is $O\left(n^3
\log^{2.5}(n\|A\|) \right)$. In the worst case (step 2) we can
choose between the CRA algorithm and the algorithms of
\cite{Storjohann2004, Eberly2000, Kaltofen2005}. In fact, in the
expected case we do not need to run this step.
The experiment proved that thanks the adaptive solutions this algorithm performs better than other implementation for a larger group of matrices.

\section{Rational Determinant Algorithms}\label{sec:rda}


The algorithms to compute the rational determinant are based on the
ideas described in section \ref{sec:ea}. We present four main
strategies to compute the rational determinant. They all use CRA
which allows us to compute the determinant of the matrix modulo a
product $p_1\cdots p_k$ of primes. Then the first variant uses the
rational reconstruction to obtain the rational result. In order to
make use of Early Termination condition we have to precondition the
determinant to obtain its integer multiplication. Preconditioning of
the matrix allows us to use the integers determinant algorithms. The
application of two determinant algorithms is studied here. The
common requirements for all algorithm are shown in \ref{alg:req}.
The algorithms are Monte Carlo type due to the early termination
used.


\floatname{algorithm}{Requirements}
\setcounter{algorithm}{-1}
\begin{algorithm}\caption{}\label{alg:req}
\begin{algorithmic}
\REQUIRE $A$ - an $m\times m$ rational matrix; \REQUIRE $D_i$,
$i=1\dots m$ - the common denominator of the entries of the $i$th
row (column); \REQUIRE $N,D$ - the bounds for the numerator and the
denominator of $\det(A)$, $D=\pi D_i$;
\REQUIRE A set $P$ of random primes;
\REQUIRE $Ratrec(a,b,u,M, N, D)$ - a procedure which reconstructs $\frac{a}{b}=u\mmod M, a <N, b < D$ or returns FAIL.
\ENSURE $\det(A)$ - the determinant of the matrix.
\end{algorithmic}
\end{algorithm}

\floatname{algorithm}{Algorithm} The effectiveness of our methods
depends heavily on the number of modular determinants computed and
thus on the bound $N$ and $D$ for the numerator and the denominator
of the determinant. One can compute $D$ as the product of lcm of all
denominators in a row (or a column). Then $N$ can be computed as
$D\cdot H$, where $H$ is the Hadamard bound for matrix $A$. One
should notice that the bounds can be largely overestimated. Thus, we
proposed output-dependant approach which allows us to reduce the
number of iteration.

The first idea is to employ the CRA scheme and compute the
determinant for the modular images of a rational matrix. In the case
when the determinant is rational, early termination condition never
holds. Instead, we have to compute the bounds $D$ and $N$ for the
denominator and numerator of the determinant.
As soon as the product of primes $M=p_1\cdots p_k$ overcomes $2ND$ we can apply
rational reconstruction and reconstruct the determinant from the
modular image. We can also use an output dependent rational
reconstruction as described in section \ref{sec:rr}. This strategy is
presented as algorithm \algrdRat.  An early termination in the
rational case would required applying the rational reconstruction from
time to time with the bounds $N=D=\sqrt{\frac{M_k}{2}}$ and wait for the result
to re-occur. This leads to solution when $M > 2\max\{n^2, d^2\}$,
where $n,d$ are the numerator and denominator of the determinant.

\begin{algorithm}\caption{\algrdRat}
\begin{algorithmic}[1]
\STATE $i=0, k=0, n=0, d=1, M = 1, u=0;$
\REPEAT
\STATE ++$i$; Get $p_i$ from $P$;
\STATE Compute $A_i = A \mmod p_i$;
\STATE Compute $u_i = \det(A_i)$;
\STATE Reconstruct $u= \det(A) \mmod M p_i$ using $M, u, u_i, p_i$, $M = M p_i$;
\IF {$i =k^2$}
    \STATE s = Ratrec($n,d,u,M, \sqrt{\frac{M}{2}},\sqrt{\frac{M}{2}}$);++$k$;
    \IFTHENEND{s $\neq$ FAIL}{return $n,d$;}
\ENDIF
\UNTIL{$M > 2ND$}
    \STATE status = Ratrec($n,d,u,M,N,D$);
    \IFTHENEND{status $\neq$ FAIL}{return $n,d$;}
\end{algorithmic}
\end{algorithm}

The second method can use the denominator bound~$D$ to make the
CRA loop look for an integer value. Again, we compute the modular
image of a rational matrix $A$ but this time we call CRA to
look for $D\times\det(A)$ which is integer. Now the classic ET
condition can be used and the result is obtained as soon
as $M > n \frac{D}{d}$. The effectiveness of this method depends
therefore on the exactness of denominator bound $D$. Experimental
results show that it is sufficient in practice, see sec. \ref{sec:exp}
table \ref{tab:app}. This strategy is presented as algorithm \algrdDet.

\begin{algorithm}\caption{\algrdDet}
\begin{algorithmic}[1]
\STATE $i=0;M=1; u=0$; \REPEAT \STATE ++$i$;Get $p_i$ from $P$;
\STATE Compute $A_i = A \mmod p_i$; \STATE Compute $u_i =
D\cdot\det(A_i)$; \STATE reconstruct $u= D\cdot\det(A) \mmod M 
p_i$ using $M, u, u_i, p_i$, $M = M \cdot p_i$ \IFTHENEND {ET
holds}{return $\frac{u}{\gcd(u,D)},\frac{D}{\gcd(u,D)}$;} \UNTIL{$M
> 2ND$} \STATE return $\frac{u}{\gcd(u,D)},\frac{D}{\gcd(u,D)}$;
\end{algorithmic}
\end{algorithm}

The last two strategies require an integer matrix $\tilde{A}$ which can be obtained by preconditioning the rational matrix $A$. In order to obtain an integer matrix, the easiest way would be to take matrix $A'=D(A)A$, where $D(A)$ is the common denominator of all entries. In the general case, where the entries of $A$ are fractions $\frac{a_{ij}}{b_{ij}}$ with numerator and denominator bounded by $\|A\|$, this is not the best choice as the size of $D(A)$ can be as large as $O(\|A\|^{m^2})$. This causes $\log(\|A'\|)$ to be $O(m^2)$. Moreover, the denominator approximation is $D(A)^m$ in this case, which is $O(m^3)$ in size. We have already defined a tighter bound for the denominator of $\det(A)$ by $\pi D_i$, which is $O(m^2)$ in size. Now, if we want to use the integer matrix $\tilde{A}$ then we can precondition $A$ by taking $\tilde{A}=A\diag(D_i)$, where $D_i$ are the common denominators of the rows (or $\tilde{A}=\diag(D_i)A$, where $D_i$ are the common denominators of the columns).
For the preconditioned matrix $\tilde{A}$ all integer determinant
algorithms can be applied. In particular the hybrid determinant
algorithm of \cite{jgd:2006:det} can be used. The drawback of this
approach is the size of the coefficients of $\tilde{A}$ compared to
$A$, see section \ref{sec:exp} table \ref{tab:him}. This forced us
to use early terminated rational reconstruction for system solving
in the Dixon $p$-adic lifting algorithm. The strategies that use the
CRA loop or the hybrid algorithm are presented as algorithms
\algrdMat~ and \algrdDixon~ respectively.

\begin{algorithm}\caption{\algrdMat}
\begin{algorithmic}[1]
\STATE $i=0;M=1; u=0;$ \STATE Compute $A=A\diag(D_i)$ (or
$\diag(D_i)A)$ \REPEAT \STATE Get $p_i$ from $P$; \STATE Compute
$A_i = A \mmod p_i$; \STATE Compute $u_i = \det(A_i)$; \STATE
reconstruct $u= \det(A) \mmod M p_i$ using $M, u, u_i, p_i$, $M
= M\cdot p_i$ \IFTHENEND {ET holds}{ return
$\frac{u}{\gcd(u,D)},\frac{D}{\gcd(u,D)}$;} \UNTIL{$M > 2ND$} \STATE
return $\frac{u}{\gcd(u,D)},\frac{D}{\gcd(u,D)}$;
\end{algorithmic}
\end{algorithm}

\begin{algorithm}\caption{\algrdDixon}
\begin{algorithmic}[1]
\STATE Compute $A=A\diag(D_i)$ (or $\diag(D_i)A)$;
\STATE Compute $u=\det(A)$ by HybridDet \cite{jgd:2006:det};
\STATE return $\frac{u}{\gcd(u,D)},\frac{D}{\gcd(u,D)}$;
\end{algorithmic}
\end{algorithm}









\section{Complexity Analysis}\label{sec:cpl}

In this section we study the complexity of the algorithms presented in section \ref{sec:rda}. In subsection \ref{sec:gen} we present the analysis of the general case, where we assume that the entries of the matrix are fractions with numerators and denominators bounded by $\|A\|$. Then, in subsection \ref{sec:spc}, we will focus on two special cases i.e. matrices of decimal fractions and Hilbert matrices.


The complexity of the strategies described in section \ref{sec:rda}
depends on the number of iterations required by the {\bf while}
loop of CRA. Then, depending on the strategy, we have to include the
cost of computing the homomorphic image of the matrix, the cost of the rational reconstruction or the cost of $p$-adic lifting. If we use the early termination condition, the number of steps required for the computation of $\det(A)$ depends on the values: $m$ - the size of the matrix, $n,d$ - the real values of the numerator and denominator of $\det(A)$ and $D$ - the bound for the denominator. The cost of homomorphic imaging depends on the maximum norm of the matrix i.e. $\|A\|=\max\{\|a_{ij}\|,b_{ij}\}$ and $\|\tilde{A}\|$.

\subsection{General case}\label{sec:gen}
We start this section by the analysis of the  rational homomorphic imaging schemes. We have the following lemma.

\begin{lemma}\label{lem:him}
Let $p$ be a word-size prime. Then the complexity of computing the
modu\-lar image at $p$ for a rational matrix $A$ is
$O(m^2(\log(\|A\|)) + \EEA(p))$ word operations.
\end{lemma}
\begin{proof}
For a matrix without a pattern we compute an image for all $m^2$
entries. For a rational fraction the cost is $O(\log(\|A\|))$ for the
computation of the modular image of the numerator and denominator and
$\EEA(p)=O(\log(p)\log(\log(p)))$ for the modular inverse computation by fast
extended Euclidean algorithm. Therefore for a word-size $\|A\|$ the cost of
computing the image is $O(1)$ yet important, due to the constant for
computing the inverse of an element mod $p$.
\end{proof}

For the integer case, the cost is $\log(\|\tilde{A}\|))$.
We can notice that $\log(\|\tilde{A}\|)$ can be $O(m\log(\|A\|))$ in the worst case, so the complexity of homomorphic imaging in terms of $m$ is $O(m^2)$ for the rational and $O(m^3)$ in the integer case. But if $\|\tilde{A}\| < p$ the cost of imaging for one element is $1$. 
Thus, if both $\|\tilde{A}\|$ and $\|A\|$ are less than $p$, the complexity of the homomorphic imaging becomes $m^2\EEA(p)$ for the rational and $m^2$ for the integer case. In this case, it is better to use integer imaging. On the other hand, if matrix $A$ is structured, for example it is Hankel-type, we have the complexity $m\EEA(p)$ for rational imaging. Due to the preconditioning, we loose the structure pattern for $\tilde{A}$ and the complexity of integer imaging rests without change. Finally we notice, that for sparse matrices with $\Omega$ elements, we can take $\Omega$ instead of $m^2$ in the complexity formula.

Putting it together we have the following theorem.

\begin{theorem}\label{thm:com}
The worst case complexities of the strategies for computing the determinant of a rational matrix $A$ of size $m$ are
\begin{enumerate}[1.]
\item $O\left(k(m^2\log(\|A\|) + m^{\omega})\right) + O^{\sim}\left(k\sqrt{k}\right) $
    for \algrdRat, where $O^{\sim}$ hides some $\log(k)$ factors;
\item $O\left(\log(\frac{D}{d}n)(m^2\log(\|A\|) + m^{\omega})\right)$
for \algrdDet;
\item $O\left(\log(\frac{D}{d}n)(m^2\log(\|\tilde{A}\|) + m^{\omega})\right)$
for \algrdMat;
\item $O^{\sim}(x(m^2(\log(m)+\log(\|\tilde{A}\|)) +
mx^{\frac{1}{2}})+O(\log(\frac{D}{d}\frac{n}{s_m}+1)(m^2\log(\|\tilde{A}\|) +
m^{\omega}))$ for \algrdDixon, where $s_m=s_m(\tilde{A})$ and $x\in m(\log(m\|\tilde{A}\|\|b\|)$ is
the size of solution to $\tilde{A}x=b$.
\end{enumerate}
Here $\tilde{A}$ is equal to $A\diag(D_i)$ as in  section
\ref{sec:rda}; $n$, $d$ are the numerator and denominator of
$\det(A)$ and $k=O(\max(\log(n), \log(d)))$.
\end{theorem}

\begin{proof}

The complexities can be obtained by a careful examination of the number of CRA
steps.
The result for alg. \algrdRat~ takes into account the cost of the rational reconstruction which is
performed at most $O(\sqrt{k})$ times. In alg. \algrdDixon~ we
introduce $x$ to estimate the cost of early terminated $p$-adic
lifting. The size of $x$ can generally vary depending on the choice of $b$ but
is $O(m\log(m\|\tilde{A}\|\|b\|))$ in the worst case. To further
evaluate the worst case complexity of alg. \algrdDixon~ we assumed
that {\em HybridDet} continues to use CRA loop in the worst case. Thus
the number of iterations $O(\log(\frac{D}{d}\frac{n}{s_m}))$ and the complexity.
\end{proof}

Special care should be taken if we consider the use of alg. \algrdDixon.
As $\|\tilde{A}\|$ can potentially be $O^{\sim}(m)$ in size and with a
pessimistic bound on $x$, its worst case complexity can be $O^{\sim}(\log(m^4))$, which is worse than for the CRA computation. Nevertheless, the gain of computing $s_m$  can be important, as it is the case in the {\em HybridDet} algorithm, see \cite{jgd:2006:det}.

\subsection{Complexity in the special cases}\label{sec:spc}
By the precedent remarks it should be visible, that the analysis of the strategies should be divided into two main cases. One would consist of the matrices, whose entries are given by decimal fraction, or more generally, where the common denominator of all entries, the common denominator of the rows and the norm of $A$ are of the same order i.e. $D(A)=O(D_i)=O(\|A\|)$. In the other case matrix entries are given as fractions with different denominators. We will study the complexity of the algorithms on the example of Hilbert matrices.

In the case of matrices of decimal fractions let us further assume that $\|A\|$ is $O(1)$. This would be the case of numerous ill-conditioned matrices emerging from different applications in science and engineering. In order to better describe the differences between the algorithms, we include the cost of EEA when it is relevant. The theorem is a straightforward consequence of theorem \ref{thm:com}.
\begin{theorem}
The complexities of the strategies in the case when $\|A\|=O(\tilde{\|A\|})=O(1)$ are:
\begin{enumerate}
\item $O^{\sim}\left(k(m^2\EEA(p) + m^{\omega} + k\sqrt{k}\right) $
    for alg. \algrdRat;
\item $O^{\sim}\left(\log(\frac{D}{d}n)(m^2\EEA(p)+ m^{\omega})\right)$
for alg. \algrdDet;
\item $O^{\sim}\left(\log(\frac{D}{d}n)(m^2 + m^{\omega})\right)$
for alg. \algrdMat;
\item $O^{\sim}\left(x(m^2\log(m) + mx^{\frac{1}{2}} ))+
\log(\frac{D}{d}\frac{n}{s_m})(m^2 + m^{\omega})\right)$
    for alg. \algrdDixon.
\end{enumerate}
where $k, x$ are as in theorem \ref{thm:com}.
\end{theorem}
The analysis suggests that the algorithm \algrdMat~ should be better than
\algrdDet~\-(see \ref{lem:him} for the homomorphic image
complexity). The performance analysis in section \ref{sec:exp}
confirms this observation. Furthermore, as long as the Smith form of
$\tilde{A}$ is simple, we encourage the use of strategy \algrdDixon.
In particular, we can establish an equivalence between matrices $A$
of random decimal fractions with $e$ decimal places taken randomly
an uniformly from the interval $[0,1]$ and matrices $\tilde{A}$,
$\|\tilde{A}\| < 10^e$. This allows us to use the expected
complexity of the hybrid algorithms of \cite{jgd:2006:det} as the
expected complexity of the rational determinant computation by alg.
\algrdDixon.
Also,
the preconditioning should be used instead of strategy \algrdRat. For more details see section \ref{sec:exp}.

The other group consists of matrices with rational entries given by
fractions with very different denominators. As a model case we can
consider Hilbert matrices. Hilbert matrices are the matrices of the
form $H_m = [\frac{1}{i+j-1}]_{i,j=1..m}$. They are benchmarks examples for
many numerical methods. The formula for the determinant of a Hilbert matrix is
well known and is given by the equation
$$
 \frac{1}{\det(H_m)}=\Pi_{k=1}^{m-1} (2k+1) \left( \begin {matrix}2k\\k\end{matrix}\right)^2.
$$


\begin{theorem}\label{thm:hil}
The complexities for rational determinant strategies in the case of Hilbert matrices are
\begin{enumerate}
\item $O\left(m^2\log(m)(m^{\omega} + m\sqrt{\log(m)})\right) $
    for alg. \algrdRat;
\item $O\left(m^{\omega+2}\log(m)\right)$
for alg. \algrdDet;
\item $O\left(m^5)\log(m)\right)$
for alg. \algrdMat;
\item $O(s_mm^3\log^2(m)+m^5\log(m))$
    for alg. \algrdDixon.
\end{enumerate}
\end{theorem}
\begin{proof}
One should notice that $\log(\frac{1}{\det(H_m)})$ is $O(m^2\log(m))$. The size of entries of $H_m$ is $\log(\|H_m\|)=O(\log(m))$ and $\log(\|\tilde{H_m}\|)=O(m\log(m))$.
\end{proof}

In the case of Hilbert matrices algorithm \algrdDet~\-has the best time complexity and also performed best in the experiments, see section \ref{sec:exp}. Since the numerator is equal to $1$, we only have to recover the size of the over-approximation. Experimental results show, that its size is equal to about 8\% of the denominator size. Therefore, alg. \algrdDet, \algrdMat~ perform about 25 times less iterations than \algrdRat. As for the algorithm \algrdDixon, the study of the Smith form of $\tilde{H_m}$ has revealed that it is quite complex, with about $2\sqrt{m}$ nontrivial factors and the size $\log(s_m(\tilde{H_m}))$ equal $O(m)$. Thus, it is not worth computing \algrdDixon~ due to the high cost of the algorithm and poor gain.


\section{Performance comparison}\label{sec:exp}

In this section we present the experimental results for four strategies
from section~\ref{sec:rda}. We have tested the performance of four strategies on three
matrix sets: random, ill-conditioned and Hilbert matrices. 

We generated the random matrices using Matlab procedure {\em rand}. The entries of the matrices are decimal fractions with 6 decimal places chosen randomly from the interval $[0,1]$. The determinant  of the resulting matrices is large in the absolute value. The result of the numerical procedure of  Matlab is $\pm\infty$.

Ill-conditioned matrices have been chosen from the Matrix Market \cite{matrixMarket} Harwell-Boeing collection. We chose three sets: Grenoble, Astroph and Bcsstruc3.
Grenoble set represents the results of the simulation of computer systems.
The sizes of the matrices varies from 115 to 1107 and the condition numbers range from $1.5\cdot 10^2$ in the case of the smallest matrix to $9.7\cdot 10^7$ for the biggest.
The decimal precision of the entries depends on the matrix and ranges from 1 to 5 decimal places.
The determinants are close to 0.
For these matrices, Matlab procedure {\em det} computes the result correctly up to the $5$th decimal place.
Since matrix entries seem to be represented as rounded expansions of rational numbers, we computed the determinant of the matrices ''as is'' and then we took continued fractions approximants of the entries with the same precision as the decimal fractions.

Astroph set describes the process of nonlinear radiative transfer and statistical equilibrium in astrophysics.
The condition number is $3.6\cdot 10^{17}$ for the small $180\times 180 $ matrix and $1.7\cdot 10^{14}$ for the $765\times 765$ one.
The result of Matlab computation is $-\infty$.
Bcsstruc3 gives dynamic analyses in structural engineering.
All matrices are symmetric.
The condition number is about $10^{11}$ for matrices 19 and 20 and $10^{5}$ for matrix 22.
The result of Matlab computation is $\infty$.
%
%
%
%

We split the analysis of the performance of the algorithms in three phases. First, we will consider the cost of rational-modular vs. integer-modular imaging and compare it with the results for $\|A\|$ and $\|\tilde{A}\|$. 
Then we will take a look on the numerator and denominator approximations $D$ and $N$ computed by our algorithms. Finally, we give the timings for all strategies and compare their performance.

As we can see in table \ref{tab:him}, the time of computing an integer image can be several times shorter than for the rational image provided that the size of preconditioned matrix is still small. This is not the case for Hilbert matrices of dimension $\geq 250$ , when the time of rational image computation is better. Furthermore, for structured matrices, like Hilbert, we can reduce  the number of images computed. For a Hankel-type matrix, there are only $2n-1$ images to compute, which makes the cost of imaging negligible.

\begin{table}\centering
\begin{tiny}
\begin{tabular}{|@{}c@{}|@{}c@{}|@{}c@{}|@{}c@{}|@{}c@{}|@{}c@{}|}
\hline
A & RatIm & IntIm & IntIm/RatIm & $\log(\|A\|)$ & $\log(\|\tilde{A}\|)$\\\hline
bccstk817   &0.14587    &0.03126    &4.66696 & 60&66\\\hline
bccstk485   &0.05189    &0.01123    &4.61980 & 65&69\\\hline
bccstk138   &0.00280    &0.00050    &5.53681 & 42&42\\\hline
mmca180     &0.00808    &0.00120    &6.74795 & 77&76\\\hline
mccf765     &0.13222    &0.03215    &4.11219 & 70&68\\\hline
grenoble115 &0.00162    &0.00019    &8.51887 & 19&19\\\hline
grenoble185 &0.00746    &0.00096    &7.8125  & 19&19\\\hline
grenoble216a&0.01055    &0.00145    &7.25    & 2&1\\\hline
grenoble216b&0.0105     &0.00106    &9.90055 & 19&19\\\hline
grenoble343 &0.0264     &0.00507    &5.21053 & 2&1\\\hline
grenoble512 &0.0588     &0.0126     &4.66667 & 2&1\\\hline
grenoble1107&0.26762    &0.05682    &4.70958 & 16&16\\\hline
random200   &0.037      &0.003      &11.692  & 19&19\\\hline
random500   &0.330      &0.028      &11.831  & 19&19\\\hline
random800   &0.599      &0.071      &8.436   & 19&19\\\hline
random1000  &0.934      &0.111      &8.452   & 19&19\\\hline\hline
hilbert100  &0.00414    &0.00255    &1.62264 & 7&289\\\hline
hilbert200  &0.02174    &0.01984    &1.09552 & 8&567\\\hline
hilbert250 &0.03481    &0.03629    &0.95942 & 8&714\\\hline
hilbert300  &0.05093    &0.05967    &0.85350 & 9&847\\\hline
hilbert400  &0.09307    &0.13343    &0.69756 & 9&1134\\\hline
hilbert600  &0.21485    &0.41759    &0.51450 & 10&1711\\\hline
hilbert800  &0.38839    &0.94920    &0.40917 & 10&2294\\\hline
hilbert1000 &0.61425    &1.81285    &0.33883 & 10&2866\\\hline
\end{tabular}\caption{Comparison of the times (in seconds) for homomorphic imaging are given in columns RatIm (for rational) and IntIm (for integer). The ratio of the timings is given in column 3. Last two columns give the size of entries for $A$ and $\tilde{A}$. Matrix size is included in its name.}\label{tab:him}
\end{tiny}
\end{table}

The performance of the algorithms depends on the accuracy of denominator approximation used. For the bound $D=\pi D_i$, the resulting size of the over-approximation is shown in table \ref{tab:app}, column 4. In algorithm \algrdDixon~ we additionally approximate the numerator by computing $s_m(\tilde{A})$. In this case we are interested in the value $App(N)=s_m(\tilde{A})$ and $\frac{D}{d}\frac{n}{s_m(\tilde{A})}$ which we compute instead of the numerator. As we can see in the table, the quality of the approximation of the denominator depends on the matrix and ranges from 1-2\% in the case of sparse matrices in the Grenoble set, to 80\% for Bccstk matrices. For Hilbert matrices the approximation is quite efficient, the over-approximation is always less than 10\%. Table \ref{tab:res} shows that despite the size of the over-approximation, preconditioning allow us to gain enough to beat the naive \algrdRat~ algorithm. If the size of $\|\tilde{A}\|$ is small, as is the case for sparse matrices, we can compute $s_m(\tilde{A})$ at a  relatively low cost and efficiently approximate the numerator.



\begin{table}\centering
\begin{tiny}
\begin{tabular}{|@{}c@{}|@{}c@{}|@{}c@{}|@{}c@{}|@{}c@{}|@{}c@{}|@{}c@{}|}\hline
$A$ & $\log(d)$ &$\log(n)$ &$\log(D/d)$& $\frac{\log(D/d)}{d}$ &$\log(App(n))$ & $\log(\frac{Dn}{dApp(N)})$\\\hline
bccstk817   &7845   &36169  &6294   &0.802  &25923  &16540\\\hline
bccstk485   &3903   &21921  &2538   &0.650  &16225  &8234\\\hline
bccstk138   &2576   &5040   &139    &0.054  &3880   &299\\\hline
mmca180        &1663   &7341   &571    &0.343  &7375   &537\\\hline
mccf765        &5503   &32451  &2626   &0.477  &32483  &2594\\\hline
grenoble115 &2243   &2136   &36     &0.016  &1526   &646\\\hline
grenoble185 &3072   &2785   &3      &0.001  &2777   &11\\\hline
grenoble216a&423    &131    &9      &0.021  &124    &16\\\hline
grenoble216b&4110   &3278   &193    &0.047  &683    &2788\\\hline
grenoble343 &678    &209    &8      &0.012  &201    &16\\\hline
grenoble512 &1009   &303    &15     &0.015  &306    &12\\\hline
grenoble1107&15639  &14002  &2707   &0.173  &7184   &9525\\\hline
random200   &3986   &4255   &0      &0      &4255   &0\\\hline
random500   &9961   &10952  &4      &0      &10956  &0\\\hline
random800   &15944  &17797  &1      &0      &17798  &0\\\hline
random1000  &19931  &22407  &0      &0      &22404  &3\\\hline\hline
hilbert100  &19737  &1      &1690   &0.086  &130    &1561\\\hline
hilbert200  &79472  &1      &6493   &0.082  &290    &6204\\\hline
hilbert300  &179207 &1      &14323  &0.080  &424    &13900\\\hline
hilbert400  &318942 &1      &26509  &0.083  &563    &25947\\\hline
hilbert600  &718412 &1      &59948  &0.083  &848    &59101\\\hline
hilbert800  &1277881&1      &103581 &0.081  &1133   &102449\\\hline
hilbert1000 &1997351&1      &164550 &0.082  &1424   &163127\\\hline
\end{tabular}\caption{The size of the numerator $n$ and denominator $d$ of $\det(A)$, the size of the denominator over-approximation $D/d$ computed by \algrdDet~ and \algrdMat; the numerator approximation $App(n)$ obtained as $s_m$ in \algrdDixon, and the size of the part remaining to compute. $s_m$ depends on $n$ and the over-approximation $D/d$.}\label{tab:app}
\end{tiny}
\end{table}
\begin{table}\centering
\begin{tiny}
\begin{tabular}{|@{}c@{}|@{}c@{}|@{}c@{}|@{}c@{}|@{}c@{}|}
\hline
Matrix  &\algrdRat  &\algrdDet  &\algrdMat  &\algrdDixon\\\hline
bccstk817   &*           &789.02    &553.624&{\bf 318.62}\\\hline
bccstk485   &278.964    &143.888&95.836 &{\bf 57.144}\\\hline
bccstk138   &4.12       &1.868  &1.324  &{\bf 0.764}\\\hline
mmca180    &14.404 &5.896  &3.644  &{\bf 1.604}\\\hline
mccf765    &*  &585.724    &416.352    &{\bf 128.24}\\\hline
grenoble115 &1.444  &0.591813   &0.456  &{\bf 0.288}\\\hline
grenoble185 &5.86   &2.34   &1.456  &{\bf 0.468}\\\hline
grenoble216a    &1.052  &0.268  &{\bf 0.248}    &0.26\\\hline
grenoble216b    &10.448 &3.852  &{\bf 2.204}    &2.128\\\hline
grenoble343 &4.292  &0.924  &0.832  &{\bf 0.732}\\\hline
grenoble512 &14.844 &2.868  &2.48   &{\bf 1.072}\\\hline
grenoble1107    &*  &698.436    &519.368    &{\bf 367.448}\\\hline
random200   &24.096     &10.776  &3.996     &{\bf 2.980}\\\hline
random500   &432.448    &180.448 &71.492    &{\bf 54.996}\\\hline
random800   &1715.316   &789.154 &331.008   &{\bf 205.188}\\\hline
random1000  &*          &1572.024&662.956   &{\bf 403.232}\\\hline\hline
hilbert100  &17.860 &0.664  &{\bf 0.548}    &0.712\\\hline
hilbert200  &330.280    &11.104 &{\bf 10.52}    &11.312\\\hline
hilbert300  &*  &{\bf 59.144}   &65.236 &66.872\\\hline
hilbert400  &*  &{\bf 200.844}  &252.676    &265.276\\\hline
hilbert600  &*  &{\bf 1072.754} &1664.738   &1735.574\\\hline
hilbert800  &*  &{\bf 3476.188} &6299.98    &8830.372\\\hline
hilbert1000 &*  &{\bf 8870.534} &18466.348  & 19328.66
\\\hline
\end{tabular}\caption{Timing comparison for 4 rational determinant
strategies. All times in seconds. Best times in bold.}\label{tab:res}
\end{tiny}
\end{table}
The timings for all algorithms are shown in table \ref{tab:res}. The
results for Hilbert matrices agree with the complexity estimation in
Thm. \ref{thm:hil}. Note that alg. \algrdDixon~ is usually the best
for the matrices from MatrixMarket collection. 

For the Grenoble set, the approximation by continued fractions allowed quite well, in our opinion, to reconstruct the orginal rational matrix connected to the problem. Despite the difference in properties, the running times for the decimal and continued fractions variants were simmilar. However, although the matrices were close in the maximum norm, the determinants ratio reached as much as 2 in the case of grenoble1107.   

In figure \ref{fig:hil} we present the results of the determinant computation for Hilbert matrices. We compare the timings for algorithm \algrdRat,  \algrdDet, \algrdMat, \algrdDixon, and the Maple LinearAlgebra::Determinant algorithm with {\em method=rational}. The best performance is observed for a variant of algorithm \algrdDet~ which takes into account the Hankel structure of the matrix.

\begin{figure}\label{fig:hil}
\centering
\includegraphics[width=200pt]{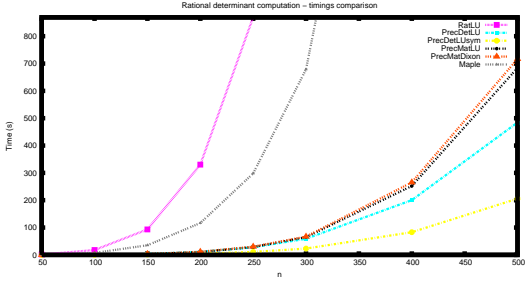}
\caption{Comparison of the timings for the exact computation of the
rational determinant of Hilbert matrices. The results for algorithms
\algrdRat, \algrdDet, \algrdMat~ and \algrdDixon~ implemented in
LinBox and Maple Determinant procedure are shown. Algorithm
\algrdDet~ is used in the classic and symmetric variant, which takes
into account the Hankel structure of the matrix. All times in
seconds. }
\end{figure}




\section{Conclusions}\label{sec:ccl}

It this paper we have presented four strategies for exact computation of the determinant of a rational matrix. We have evaluated the performance of these algorithms on several sets of matrices. The performance of the algorithms suggests that there exists a clear division between the matrices given as a rational approximation (by decimal fractions) of real valued matrices and the matrices with a great diversity of the denominators of the entries. For the first case, matrix preconditioning which leads to a integer matrix is proposed, which allows us to use integer determinant algorithms, see solution \algrdDixon. For the second case, determinant preconditioning is preferred, which does not lead to matrix coefficient blow-up. In general, preconditioning proved more useful than rational reconstruction tools, although better early termination methods where the modulus $M$ is linear in the size of the output $n$ and $d$ can bring a change, see \cite{kho-mon:2006,ol-sto:2006}.

An adaptive solution should 
be able to choose the best storage method and homomorphic imaging scheme,
and work independently of the determinant over-approximation.

We propose the following solution, which incorporates the elements of all algorithms 
\begin{enumerate}
\item Compute $D=\pi D_i$, $\tilde{A}$; set $N=1$;
\item If $\log_p(\|\tilde{A}\| < C)$ compute $N=s_m(\|\tilde{A}\|)$ - see alg. \algrdDixon
\item Compute the modular image of the rational matrix $A$ and integer matrix $\tilde{A}$, determine whether to use \algrdDet~ or \algrdMat~ based on the timings.
\item Run the ET CRA loop for $\frac{D}{N}\cdot \det(A)$ using \algrdDet~ or \algrdMat.
\item From time to time check by rational reconstruction the early termination condition on $\det(A)$ - see \algrdRat.
\end{enumerate}
This algorithm can be further developed to compute other invariant factors as in alg. \algrdDixon~ if relevant. Notice, that the cost of introducing solution \algrdRat~ to the adaptive algorithm is virtually that of rational reconstruction.

Further work can include intertwining algorithms \algrdRat~ and \algrdDet~ to include the use of less exact determinant preconditioners, which potentially are not a multiple of $d$. The aim would to reduce a factor of the denominator by preconditioning and reconstruct the remaining part by rational reconstruction. The strategy should be effective, if the over-approximation caused by preconditioning is reduced but a large fraction of the denominator is obtained at the same time. 
 For example, $D=\pi D_i /gcd(D_i)$ could be considered.


Further work can then focus on the implementation of the solution in the case of sparse matrices and on the parallelization of the algorithms.

In this paper we have considered the case of dense matrices in the analysis of the complexity of the strategies as well as in the implementation. However, sparse matrix counterparts of the algorithms can also be used. For the modular determinant computation one could used the algorithm of Wiedemann \cite{Wiedemann:1986:SSLE} that computes the determinant by finding the characteristic polynomial of the matrix. In alg. \algrdDixon~ the sparse solver of \cite{EbGies2006} can be used.

The strategies described in this paper contain elements that allow parallelization. This concerns in particular the CRA loop, where several iterations can be performed at the same time, see \cite{Dumas2000}. The question of an optimally distributed early termination in the case of integer Chinese reconstruction (alg. \algrdDet, \algrdMat, \algrdDixon) as well as the rational reconstruction (alg. \algrdRat) has not yet been addressed. For a parallel $p$-adic lifting for alg. \algrdDixon, see \cite{Dumas2002}.

In this paper we have developed and compared four strategies to compute the rational determinant of a matrix. We have proposed two preconditioning methods that allow us to transfer the problem from rational to integer domain. We believe that the approach described in this article can also be applied in other problems of exact computation in rational numbers such as rank computation or system solving. 

\end{document}